\documentclass[pra,twocolumn,amssymb,superscriptaddress,showpacs]{revtex4}
\usepackage{amsmath}
\usepackage{amsthm}
\usepackage{graphicx}
\usepackage[dvips]{epsfig}
\usepackage{bbm}
\usepackage{braket}
\usepackage[usenames,dvipsnames]{color}
\usepackage{wasysym}
\usepackage{bbold}

\newcommand{\ie}{\emph{i.e.}}
\newcommand{\eg}{\textit{e.g.}}
\newcommand{\tr}{\operatorname{tr}}

\newtheorem{observation}{Observation}

\begin{document}
\title{Genuine multiparticle entanglement of permutationally invariant states}

\author{Leonardo~Novo}
\affiliation{Physics of Information Group, Instituto de Telecomunica\c{c}\~oes, P-1049-001 Lisbon, Portugal}
\affiliation{Naturwissenschaftlich-Technische Fakult\"at, Universit\"at Siegen, Walter-Flex-Str.~3, D-57068 Siegen, Germany}
\author{Tobias~Moroder}
\affiliation{Naturwissenschaftlich-Technische Fakult\"at, Universit\"at Siegen, Walter-Flex-Str.~3, D-57068 Siegen, Germany}
\author{Otfried~G\"uhne}
\affiliation{Naturwissenschaftlich-Technische Fakult\"at, Universit\"at Siegen, Walter-Flex-Str.~3, D-57068 Siegen, Germany}

\date{\today}

\begin{abstract}
We consider the problem of characterizing genuine multiparticle entanglement for permutationally invariant states using the approach of PPT mixtures. We show that the evaluation of this necessary biseparability criterion scales polynomially with the number of particles. In practice, it can be evaluated easily up to ten qubits and improves existing criteria significantly. Finally, we show that our approach solves the problem of characterizing genuine multiparticle entanglement for permutationally invariant three-qubit states.
\end{abstract}

\maketitle

\section{Introduction}

Quantum state analysis of large-scale systems is non-trivial since many generic methods only work for, at most, small system sizes. Experimentally, the manipulation and control of several qubits has  already become standard, and current records comprise, for instance, entanglement between $14$ qubits in ion traps~\cite{monz11a}, ten-qubit entanglement using hyper-entangled photons~\cite{gao10} or the generation of eight entangled photons~\cite{yao11a,huang11a}. Even for such medium scale systems, the available analysis tools can be rather cumbersome like, for instance, the task of quantum state tomography, \ie, the process to determine the underlying quantum state by suitable measurements and hence gaining full information. The standard Pauli tomography scheme~\cite{james01a} scales exponentially such that a feasible, generic tomography protocol for $14$ qubits is out of scope. 

However, quite often one intends to work with special classes of states only. This offers the possibility that one can tailor or optimize the analysis tool for those more restricted sets. Such more efficient tomography protocols have been recently designed for generic states of low rank~\cite{gross10a}, particularly important low rank states like matrix product states~\cite{cramer10a} or multi-scale entanglement renormalization ansatz states~\cite{landon_cardinal12a}, or for states which possess some further symmetry like permutation invariance~\cite{Toth2010}. 

Though full information on a quantum state is appealing, it is usually dispensable since one is often more interested in a few key characteristics or properties of the states, which are also used to compare different experiments. From the plethora of interesting characteristics, the main objective of multipartite systems lies on genuine multipartite entanglement \cite{hororeview, tg09}. This is the strongest phenomenon of quantum mechanical correlations within such systems that cannot be explained via sufficient control on systems of less particle size, known as biseparable states. Despite its importance, characterization and detection of this kind of resource is still hard and only recently some methods have been proposed~\cite{guehneseevinck, huber_dicke, huber_ghzwid, vicente11a, laskowski11a, siewert}. A very promising detection method constitutes the concept of PPT mixtures~\cite{Jungnitsch2011,jungnitsch11b}, the generalization of the positive partial transpose (PPT) criterion~\cite{Peres1996} to the multipartite setting.

In this paper we tailor, similarly to the tomography protocols, the detection of genuine multipartite entanglement via PPT mixtures to permutationally invariant states. We show that the question whether a given permutationally invariant state possesses a PPT mixture requires resources which scale only polynomially in the number of qubits. Thus, in combination with the tomography protocol~\cite{Toth2010} (and its variants~\cite{dariano03a,adamson07a,klimov13a}) and its efficient state reconstruction algorithm~\cite{Moroder2012}, we develop an additional tool to analyze the data after such a quantum state tomography process. At this point, we would like to stress that the derived detection method does not rely on the fact that the underlying state indeed possesses this symmetry: If the permutationally invariant part of a quantum state is entangled, then the complete state must be entangled, too~\cite{Toth2010}. As a further result we prove that the criterion of PPT mixtures is necessary and sufficient to decide whether a given permutationally invariant three-qubit states is genuine multipartite entangled or not. Thus we obtain another interesting class of states, similar to graph-diagonal states of three and four qubits~\cite{guehne11a}, where this approach completely solves the question of genuine multiparticle entanglement. As examples, we study states like Greenberger-Horne-Zeilinger (GHZ) and Dicke states and obtain strongly improved detection statements for up to ten qubits. We would like to add that in the present paper, we focus on the numerical evaluation of the criterion of PPT mixtures. Of course, also an analytic approach via the construction of appropriate witnesses is possible. This can lead to criteria which can be used for arbitrary particle numbers. Results on this problem will be reported elsewhere \cite{marcel2013}.

The structure of this paper is as follows: Section~\ref{sec:II} summarizes the background on multipartite entanglement, the concept of PPT mixtures and permutationally invariant states. The theoretical results of this paper are given in Section~III, in particular the aforementioned results about structure and scaling of permutationally invariant PPT states and about the sufficiency statement for the three-qubit permutationally invariant case. Section~\ref{sec_sdpdetails} collects some details of our numerical implementation via semidefinite programming (SDP), which is used afterwards to test and to compare our method on special family of states in Section~\ref{sec:V}. Finally, we summarize in Section~\ref{sec:VI}.

\section{Preliminaries}\label{sec:II}

\subsection{PPT mixtures}\label{sec:IIA}
At first, let us review the concept of PPT mixtures~\cite{Jungnitsch2011}, which represents a method to detect \emph{genuine multipartite entanglement}. Similarly to the PPT criterion of the bipartite case~\cite{Peres1996}, it is a suitable relaxation on the level of quantum states which gets more tractable.


Let us first explain the method for a system of three particles because this already highlights the idea. A tripartite state is separable with respect to the bipartition $A|BC$ if it can be written as 
\begin{equation} \rho_{A|BC}^{sep}=\sum_k q_k \ket{\phi_A^k}\bra{\phi_A^k}\otimes\ket{\psi_{BC}^k}\bra{\psi_{BC}^k},
\end{equation}
where $q_k$ form a probability distribution, \ie, $\sum_k q_k=1$ and $q_k\geq 0$ for all $k$. The definition is analogous for the other possible bipartitions $B|AC$, $C|AB$. A \emph{biseparable state} is now defined as a convex combination of states which are biseparable with respect to a specific bipartition, more precisely the state can be written as
\begin{equation}
\rho_{ABC}^{bs}= p_1 \rho_{A|BC}^{sep}+p_2\rho_{B|AC}^{sep}+p_3\rho_{C|AB}^{sep}.
\end{equation}
If a state is not biseparable, it is called \emph{genuinely multipartite entangled}.

Naturally, since the bipartite separability problem is already hard, the full characterization of biseparable states can only be worse.
The idea of the PPT mixtures is to define a set of states which includes the set of biseparable states, but which is much easier to characterize than the latter. At the core of this method lies the PPT or Peres-Horodecki criterion, which was first introduced in Ref.~\cite{Peres1996}. This criterion, based on the operation of partial transposition, is a simple and powerful method to detect bipartite entanglement. If a bipartite system $\rho_{AB}$ is expanded in a chosen tensor product basis as $ \rho_{AB}=\sum_{ijkl}\rho_{ij,kl}\ket{ij}\bra{kl}$, its partial transpose with respect to the first subsystem is defined as $\rho_{AB}^{T_A}=\sum_{ijkl}\rho_{ij,kl}\ket{kj}\bra{il}$. The PPT criterion says that if a state $\rho_{AB}$ is separable, its partial transpose is positive semidefinite, \ie, it has no negative eigenvalues, or, in other words, one says the state is PPT. This implies that if $\rho_{AB}^{T_A}$ has one or more negative eigenvalues (the state is then called an NPT state), it must be entangled. Throughout this work we will often use the term positive meaning positive semidefinite.

The generalization of this criterion to the multipartite case, as introduced in Ref.~\cite{Jungnitsch2011}, goes as follows: Similarly to the definition of a separable state with respect to the bipartition $A|BC$, one can define a state $\rho_{A|BC}^{ppt}$ to be PPT with respect to that partition, and similarly for the other bipartitions. 
In analogy to the definition of a biseparable state, a PPT mixture of a three party state is defined as a convex combination of PPT states with respect to a specific bipartition
\begin{equation}
\rho^{pmix}_{ABC}= p_1\rho_{A|BC}^{ppt}+p_2\rho_{B|AC}^{ppt}+p_3\rho_{C|AB}^{ppt}.
\end{equation}
From the PPT criterion, we know that all separable states for a fixed bipartition are contained in the set of all PPT states for the same bipartition. This then implies that the set of PPT mixtures contains the set of biseparable states. Consequently, if a state is not a PPT mixture, it is genuinely multipartite entangled. In Fig.~\ref{fig:pptmixtures}, we can see a schematic representation of the set of PPT mixtures and the set of biseparable states for this three particle case.

\begin{figure}[htb]
  \centering
  \includegraphics[scale=0.22]{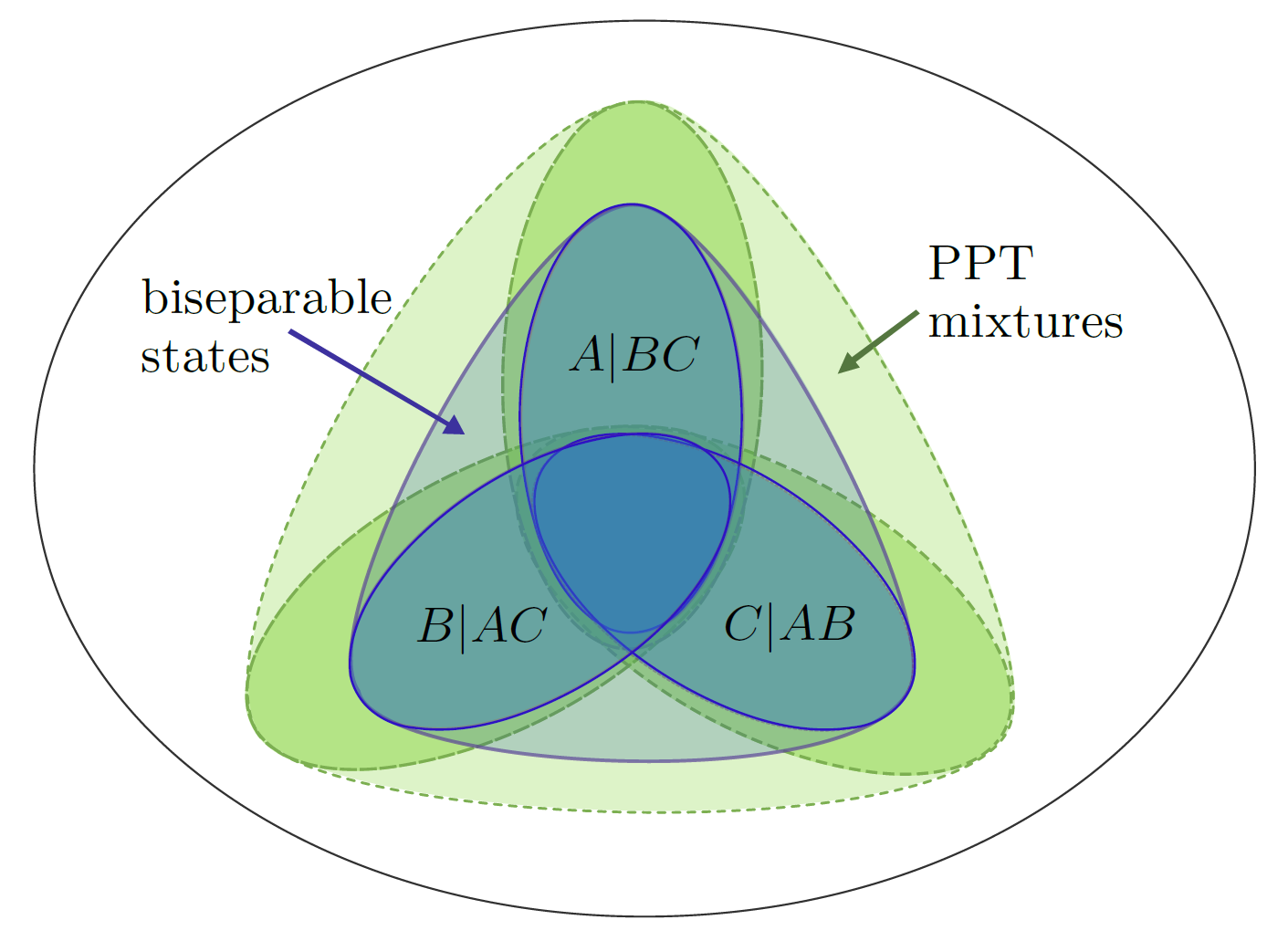}
  \caption{Here we see, for three particle states, a scheme of the biseparable states with respect to the three possible bipartitions $A|BC$, $B|AC$ and $C|AB$ (surrounded by solid blue lines). Each of these is contained in the PPT states for each respective bipartition (surrounded by dashed green lines). The larger solid blue line enveloping the other solid blue lines represents the set of biseparable states and similarly the larger dashed green line defines the set of PPT mixtures.}
  \label{fig:pptmixtures}
\end{figure}

The generalization to the $N$-partite case is straightforward. A general biseparable state of $N$ parties can be written as
\begin{equation}
\rho^{bs}=\sum\limits_{\substack{all~ bipart.\\M|\overline{M}}} p_{M|\overline{M}}\rho^{sep}_{M|\overline{M}},
\end{equation}
where the states are separable with respect to partition $M \subset \{1,\ldots,N\}$ and its complement $\overline{M}$. The sum runs over all possible bipartitions $M|\overline{M}$ of the $N$ particles.
Analogously, a PPT mixture of $N$ parties is defined by 
\begin{equation}\label{pptmixture}
\rho^{pmix}=\sum\limits_{\substack{M|\overline{M}}} p_{M|\overline{M}} \rho^{ppt}_{M|\overline{M}}=\sum\limits_{\substack{M|\overline{M}}} P_{M|\overline{M}},
\end{equation}
with operators $P_{M|\bar{M}} \equiv p_M\rho^{ppt}_{M|\overline{M}}$ that are positive and PPT with respect to $M$.

The motivation of using the concept of PPT mixture is that the problem of determining whether an arbitrary state is a PPT mixture or not can be formulated in terms of a semidefinite program (SDP) \cite{SDP06}, which makes it an easy to implement criterion to detect genuine multipartite entanglement. This SDP reads as follows
\begin{align}
\label{SDP}
{\rm min}&~~-s \\ 
\nonumber
{\rm s.t.}&~~\rho=\sum\limits_{\substack{M|\overline{M}}} P_{M|\overline{M}}, \\
\nonumber
&~~P_{M|\overline{M}}\geq s\mathbb{1},~~ {P_{M|\overline{M}}}^{T_M}\geq s\mathbb{1},~~ \mbox{for all } {M|\bar{M}},
\end{align}
where the notation $P_{M|\overline{M}}\geq s\mathbb{1}$ means that $P_{M|\overline{M}}- s\mathbb{1}$ is a positive semidefinite matrix. If the result of the optimization $s_{opt}$ is non-negative, then the state $\rho$ is a PPT mixture. Otherwise, it is genuinely multipartite entangled. Note that in Ref.~\cite{Jungnitsch2011}, the semidefinite program is written in the so-called dual form, which can be interpreted as a search for appropriate entanglement witnesses which are non-negative on all PPT mixtures. However, both forms are equivalent and thus detect the same states.

In a standard computer, though, this SDP can only be applied to generic states of up to five or six qubits~\cite{pptmixer}. In fact, it can be seen that the difficulty of this program scales exponentially with the number of qubits. Note that the number of different inequivalent bipartitions of a system of $N$ particles is given by $2^{N-1}-1$. For each of these bipartitions the Hermitian matrix $P_{M|\overline{M}}$ has $4^N$ free parameters, which are the variables of the semidefinite program (in fact, one of the matrices is fixed because of the equality constraint, but all the others are free). 

We will see in the following that if we restrict ourselves to permutationally invariant states, both the number of partitions and the number of SDP variables scale only polynomially. Furthermore, by an efficient decomposition of the matrices which have full or some permutation invariance it is possible to check the positivity conditions in terms of smaller blocks (see Sec.~\ref{sec_PIstates} and \ref{sec_sdpdetails}). This finally allows us to construct a SDP able to detect genuine multipartite entanglement for larger systems.

\subsection{Permutationally invariant states}\label{sec_PIstates}
Many experiments that aim at creating genuine multipartite entanglement, are designed in such a way that the generated state is invariant under particle interchange. Famous examples are the GHZ or the Dicke states. Mathematically, for any density matrix $\rho$ a permutationally invariant (PI) density matrix can be constructed via
\begin{equation}
\label{eq_PIstate}
\rho_{(1\dots N)}= [~\rho~]_{PI} =\dfrac{1}{N!}\sum_{\pi\in S_N}V(\pi) \rho V^\dagger(\pi)
\end{equation}
where $V(\pi)$ is a representation of the permutation $\pi\in S_N$ acting on the Hilbert space of $N$ qubits. The brackets $(1\dots N)$ should denote invariance under permutations between any of the $N$ qubits. We will sometimes employ the notation ${\left[~.~\right]}_{PI}$ to explicitly refer to this operation in order to shorten the expressions. It can be seen that Eq.~\eqref{eq_PIstate} implies that 
\begin{equation}
{\left[~\rho~\right]}_{PI}=V(\pi){\left[~\rho~\right]}_{PI}V(\pi)^{\dagger},~\mbox{for all } \pi\in S_N.
\end{equation}
A natural basis to write permutationally invariant states is given by the coupled spin basis, for which such states attain a particular simple block diagonal form,\eg,~\cite{Moroder2012,baragiola10a}. In this basis, the Hilbert space of $N$ qubits is decomposed as
\begin{equation}
\mathcal{H}=(\mathbb{C}^2)^{\otimes N}= 
\overset{N/2}{\underset{j=j_{min}}{\bigoplus}}\mathcal{H}_j\otimes \mathcal{K}_j
\end{equation}
with $j_{min} \in \{0,1/2\}$ depending on whether $N$ is even or odd. Here, $\mathcal{H}_j$ are the spin Hilbert spaces of dimension $2j+1$ and $\mathcal{K}_j$ are called the multiplicative spaces, whose dimension is given by
\begin{equation}
\dim(\mathcal{K}_j)=\left(\begin{array}{c}N\\N/2-j\end{array}\right)-
\left(\begin{array}{c}N\\N/2-j-1\end{array}\right),
\end{equation}
for $j<N/2$ and $\dim(\mathcal{K}_{N/2})=1$. The advantage of this decomposition is that any permutation $V(\pi)$ will only act non-trivially onto the multiplicative spaces, that is, they can be written as
\begin{equation}\label{eq:schur_weyl}
V(\pi)=\overset{N/2}{\underset{j=j_{min}}{\bigoplus}}\mathbb{1}\otimes \mathcal{V}_j(\pi),
\end{equation}
where $\mathcal{V}_j(\pi)$ is an irreducible representation of $S_N$ acting on $\mathcal{K}_j$~\cite{Christandl06,B1996}. This will become important shortly. Finally, we will denote the basis states by $\ket{j,m,\alpha_j}$, where $\ket{j,m}\in\mathcal{H}_j$ and $\ket{\alpha_j}\in\mathcal{K}_j$. These states $\ket{j,m,\alpha_j}$ are eigenstates of $\bf{J}^2$ and $J_z$, where ${\it{\bf{J}}}$ is the total angular momentum operator, while $J_z$ is the projection of $\bf{J}$ in the $z$ direction, with
\begin{align}
{\bf{J}}^2 \ket{j,m,\alpha_j}&=\hbar^2j(j+1)\ket{j,m,\alpha_j},\\
J_z\ket{j,m,\alpha_j}&=\hbar m\ket{j,m,\alpha_j}.
\end{align}

Any permutationally invariant state $\rho_{(1\dots N)}$ can in this formalism be written as, \eg,~\cite{Moroder2012,baragiola10a},
\begin{equation}
\label{eq:block_structure}
\rho_{(1\dots N)}=\bigoplus\limits_{j=j_{min}}^{N/2}p_j \rho_j\otimes\dfrac{\mathbb{1}_{\mathcal{K}_j}}{\dim(\mathcal{K}_j)}\\
=\bigoplus\limits_{j=j_{min}}^{N/2} B_j\otimes\mathbb{1}_{\mathcal{K}_j}, 
\end{equation}
with states $\rho_j$ of $\mathcal{H}_j$ and a probability distribution $p_j$. Complete knowledge of all probabilities $p_j$ and all states $\rho_j$ gives a complete characterization of $\rho_{(1\dots N)}$. Furthermore, we defined
\begin{equation}
B_j\equiv\dfrac{p_j\rho_j}{\dim(\mathcal{K}_j)},
\end{equation}
which are the blocks that appear in the diagonal of the PI state $\rho_{(1\dots N)}$. Each of these blocks $B_j$ appears in the diagonal exactly $\dim(\mathcal{K}_j)$ times as shown in Fig.~\ref{fig:block_structure}. From this structure, it is straightforward to compute the number of parameters needed to define a PI state given by 
\begin{equation}
\sum_{j=j_{min}}^{N/2}(2j+1)^2=\left(\begin{array}{c}N+3\\N\end{array}\right) = \mathcal{O}(N^3),
\end{equation}
which is much smaller than the $4^N-1$ parameters needed to characterize a general state. Apart from the better scaling, this block structure is also very important for the formulation of the SDP. It will be shown in Sec.~\ref{sec_sdpdetails} that all constraints of the SDP can be translated into appropriate constraints of the blocks. For instance, $\rho_{(1\dots N)}\geq 0$ is equivalent to $B_j\geq 0$ for all $j$. This is the main reason of the polynomial scaling in the end.
\begin{figure}[t]
\centering
\includegraphics[width=0.5\textwidth,trim=5cm 15cm 4cm 5cm]{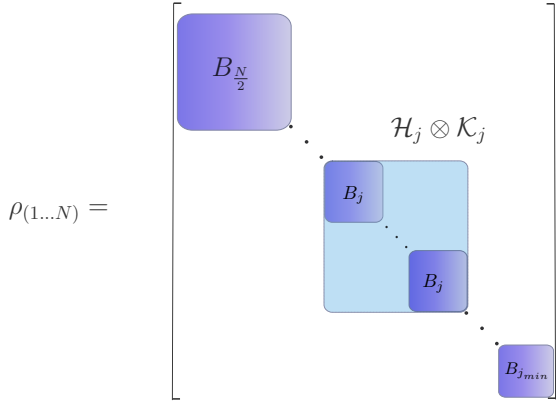}
\caption{Due to the form given by Eq.~(\ref{eq:block_structure}), a PI state has a block diagonal structure as shown above, in a suitably ordered basis. Each block $B_j$ appears $\dim(\mathcal{K}_j)$ times in the diagonal.}
  \label{fig:block_structure}
\end{figure}

Before we continue we will show how to prove Eq.~\eqref{eq:block_structure} and, at the same time, a way to carry out the operation ${\left[~.~\right]}_{PI}$ without computing the matrices $V(\pi)$. In general, this would be very hard since, for a $N$ particle system, $V(\pi)$ has dimension $4^N$ and there are $N!$ different permutations that need to be considered. Although we are always dealing with quantum states in the derivation, this block structure appears for any permutationally invariant operator. Starting with a general density matrix
\begin{equation}
\rho=\sum_{\substack{jj^\prime mm'\\\alpha_j\alpha_{j'}}}
\rho^{\alpha_j\alpha_{j'}}_{\substack{jj'mm'}}\ket{j,m,\alpha_j}\bra{j',m',\beta_{j'}},
\end{equation}   
we have
\begin{align}
&\left[~\rho~\right]_{PI}=\sum_{\substack{jj'mm'\\\alpha_j\alpha_{j'}}}
\rho^{\alpha_j\alpha_{j'}}_{\substack{jj'mm'}}
\left[\ket{j,m,\alpha_j}\bra{j',m',\beta_{j'}}\right]_{PI}
\nonumber
\\
&=\sum
\rho^{\alpha_j\alpha_{j'}}_{\substack{jj'mm'}}
\ket{j,m}\bra{j',m'}\otimes\!\sum_\pi \! \dfrac{\mathcal{V}_j(\pi)\ket{\alpha_j}\bra{\alpha_{j'}}\mathcal{V}_{j'}^{\dagger}(\pi)}{N!}
\nonumber
\end{align}
\begin{align}
&=\sum
\rho^{\alpha_j,\alpha_{j'}}_{\substack{jj'mm'}}
\ket{j,m}\bra{j',m'}\otimes\dfrac{\mathbb{1}_{\mathcal{K}_j}\delta_{jj'}
\tr(\ket{\alpha_j}\bra{\alpha_{j'}})}{\dim(\mathcal{K}_j)}
\nonumber
\\
&=\sum_{\substack{jmm'}}\left(\sum_{\alpha_j}\frac{
\rho^{\alpha_j\alpha_j}_{\substack{jjmm'}}}{\dim(\mathcal{K}_j)}\right)
\ket{j,m}\bra{j,m'}\otimes\mathbb{1}_{\mathcal{K}_j}.
\end{align}
In the first step, we used Eq.~\eqref{eq:schur_weyl}, while in the second we made use of Schur's lemma. Our proof of Eq.~\eqref{eq:block_structure} is similar to what is shown in Ref.~\cite{Moroder2012} except that here we obtain explicitly the entries of the blocks $B_j$ from the entries of $\rho$. In order to see clearly the meaning of the result obtained, we define
\begin{equation}
B_j^{\alpha_j}=\sum_{mm'}\rho^{\alpha_j\alpha_j}_{\substack{jjmm'}}
\ket{j,m}\bra{j,m'},
\end{equation}
such that, this way, the result simply reads
\begin{equation}\label{eq:block_entries}
B_j=\sum_{\alpha_j}\dfrac{B_j^{\alpha_j}}{\dim(\mathcal{K}_j)}.
\end{equation}
This means that to calculate the blocks of $\rho_{(1\dots N)}$ one has to take the average over the multiplicative spaces of the blocks of $\rho$ associated with the angular momentum~$j$.

\section{PPT mixtures of PI states}

\subsection{Characterization of PI PPT mixtures}
Here, two of the main analytical results of this work are presented, in the form of two observations. In the first one, we derive a simplified equation which characterizes a PI PPT mixture, the PPT mixture of a PI state. We show that without losing generality we can restrict the sum over bipartitions only to the ones with different number of particles on one side. Furthermore, we can impose symmetries on the unnormalized PPT states that need to be considered. The second observation proves that the number of parameters necessary to characterize a PI PPT mixture is of $\mathcal{O}(N^7)$.

\begin{observation} 
While Eq.~\eqref{pptmixture} characterizes a general PPT mixture, the equation that characterizes a PI PPT mixture can without loss of generality be written as 
\begin{equation}\label{pptmixturePI}
\rho_{(1\dots N)}^{pmix}=\sum\limits_{\substack{k=1}}^{N^\prime/2}\sum_{\pi\in S_N}V(\pi)Q_{(1\dots k)|(k+1\dots N)}V(\pi)^{\dagger},
\end{equation}
where $Q_{(1...k)|(k+1...N)}$ is an unnormalized PPT state for the partition $1...k|k+1...N$ which is additionally invariant under permutations among the first $k$ or the last $N-k$ qubits. We have $N^\prime\in \{N,N-1\}$ if $N$ is even or odd.
\end{observation}
\begin{proof}
Let $\rho_{(1\dots N)}$ be a permutationally invariant state which is a PPT mixture. Then, combining Eq.~\eqref{pptmixture} and Eq.\eqref{eq_PIstate} we can write
\begin{equation}\label{pptmixerPI2}
\rho^{pmix}_{(1\dots N)}=\sum\limits_{\substack{M|\overline{M}}}\dfrac{1}{N!}\sum_{\substack{V(\pi)\in S_N}} V(\pi)P_{M|\overline{M}} V(\pi)^\dagger.
\end{equation}
Now, let $|M|$ denote the number of elements inside the partition $M$. Then, for any bipartition $M|\overline{M}$, with $|M|=k$, there is always a permutation $\tau_M\in S_N$ which maps $M|\overline{M}$ to $1...k|k+1...N$. We define
\begin{equation}
Q_{1\dots k|k+1\dots N}=
\sum\limits_{\substack{M|\overline{M}:\\|M|=k}}V(\tau_M)P_{M|\overline{M}} V(\tau_M)^\dagger,
\end{equation}
which is a positive operator whose partial transpose of the qubits $1\dots k$ is also positive. This is true since
\begin{equation}
\left[V(\tau_M)P_{M|\overline{M}} V(\tau_M)^\dagger\right]^{T_{1\dots k}} \!= V(\tau_M) P_{M|\overline{M}}^{T_M} V(\tau_M)^\dag
\end{equation}
holds for each term in the decomposition and $P_{M|\overline{M}}$ being PPT for this bipartition.
We can now simplify Eq.~\eqref{pptmixerPI2} via
\begin{align}
\rho^{pmix}_{(1\dots N)}&
=\sum\limits_{\substack{M|\overline{M}}}\dfrac{1}{N!}\sum_{\substack{\pi\in S_N}} V(\pi \tau_M^{-1}\tau_M)P_{M|\overline{M}} V(\pi \tau_M^{-1}\tau_M)^\dagger
\nonumber
\\
&\!\!\!\!\!\!\!\!=\sum\limits_{\substack{M|\overline{M}}}\dfrac{1}{N!}\sum_{\substack{\pi'\in S_N}} V(\pi')V(\tau_M)P_{M|\overline{M}} V(\tau_M)^\dagger V(\pi')^\dagger
\nonumber
\\
&\!\!\!\!\!\!\!\!=\sum\limits_{\substack{k=1}}^{N'/2}\dfrac{1}{N!}\sum_{\substack{\pi'\in S_N}} V(\pi')Q_{1...k|k+1...N} V(\pi')^\dagger,
\end{align}
where in the first step we defined the permutation $\pi'$ as $\pi \tau_M^{-1}$. Furthermore, the $Q_{1\dots k|k+1\dots N}$ can without loss of generality be chosen to be invariant under any permutation $\pi_k\in S_k$ of the first $k$ or any $\pi_{\bar{k}}\in S_{\bar{k}}$ of the last $\bar{k}\equiv N-k$ qubits without altering the property that it is PPT. This follows since we can actively use the symmetrization as 
\begin{align}
&\sum_{\substack{\pi\in S_N}} V(\pi)Q_{1\dots k|k+1\dots N} V(\pi)^\dagger
\nonumber
\\
&=\dfrac{1}{k!} \sum_{\substack{\pi\in S_N\\\pi_k\in S_k}} V(\pi\pi_k^{-1}\pi_k)Q_{1\dots k|k+1\dots N} V(\pi\pi_k^{-1}\pi_k)^\dagger
\nonumber\\
&=\dfrac{1}{k!}\sum_{\substack{\pi'\in S_N\\\pi_k\in S_k}} V(\pi')V(\pi_k)Q_{1\dots k|k+1\dots N}V(\pi_k)^\dagger V(\pi')^\dagger
\nonumber
\\
&= \sum_{\substack{\pi\in S_N}} V(\pi)Q_{(1\dots k)|k+1\dots N} V(\pi)^\dagger,
\end{align}
where we used the notation 
\begin{align}
Q_{(1\dots k)|k+1\dots N}\!=\!\dfrac{1}{k!}\sum_{\substack{p_k\in S_k}}\!\!V(\pi_k)Q_{1\dots k|k+1\dots N}V(\pi_k)^\dagger,\!
\end{align}
and defined the permutation $\pi'$ as $\pi \pi_k$. Note that since the partial transpose acts only on the first $k$ particles this state is still PPT. Analogously, we can also symmetrize $Q_{(1\dots k)|k+1\dots N}$ for permutations of the last $N-k$ particles to obtain Eq.~\eqref{pptmixturePI}. This concludes the proof.
\end{proof}

\begin{observation}\label{obs:2}
The number of parameters needed to characterize a permutationally invariant PPT mixture is of $\mathcal{O}(N^7)$, where $N$ is the number of qubits.
\end{observation}
\begin{proof}
This scaling is due to the fact that Eq.~\eqref{pptmixturePI} exhibits two simplifications when compared to Eq.~\eqref{pptmixture}: The first one is that the number of bipartitions that need to be considered is only $N^\prime/2$ and the second is that $Q_{(1\dots k)|(k+1\dots N)}$ is permutationally invariant within each side of its respective bipartition. 

More specifically, if $\sigma^{\alpha}_{(1...k)}$ and $\sigma^{\beta}_{(k+1\dots N)}$ are operator basis elements for permutationally invariant operators of $k$ and $N-k$ respectively, any operator $Q_{(1\dots k)|(k+1\dots N)}$ can be written as
\begin{equation}
\label{decomposition}
Q_{(1\dots k)|(k+1\dots N)} =\sum_{\alpha\beta}c_{\alpha \beta}\sigma^{\alpha}_{(1...k)}\otimes\sigma^{\beta}_{(k+1...N)}.
\end{equation} 

As mentioned in Sec.~\ref{sec_PIstates} any permutationally invariant operator on $k$ qubits can be parametrized by $\mathcal{O}(k^3)$ parameters. Hence, the operator given by Eq.~\eqref{decomposition} has about $\mathcal{O}\left[(N-k)^3k^3\right]$ parameters, which, at most, can be $\mathcal{O}(N^6)$ since $k$ can be roughly $N/2$. This, together with the fact that one has to consider about $N/2$ bipartitions, leads to an overall number of parameters to describe a PI PPT mixture of $\mathcal{O}(N^7)$. This finishes this observation.
\end{proof}

\subsection{Necessity and sufficiency for PI three-qubit states} \label{necsuf3qubits}

Next we show that the method of PPT mixtures is not just only necessary, but also sufficient for biseparability of a permutationally invariant three-qubit system. Note that this result does not extend to systems of more particles where explicit counterexamples are known~\cite{toth09a,tura12a}.

\begin{observation}
A permutationally invariant three-qubit state is biseparable if and only it is a PPT mixture.
\end{observation}

\begin{proof}
Any biseparable state is also a PPT mixture as explained in the Sec.~\ref{sec:IIA}. Thus we are left to show that a PPT mixture of a three-qubit permutationally invariant state is indeed biseparable. 

For that we can without loss of generality assume the special form as given by Observation 1. Since we only have non-trivial bipartitions of $1$ vs.~$2$ particles, we obtain the following form $\rho^{pmix}_{(ABC)} = [Q_{A|(BC)}]_{\rm PI}$ or, more explicitly,
\begin{align}
\nonumber
\rho^{pmix}_{(ABC)}= \frac{1}{3} \big(  & \rho_{A|(BC)} + V_{AB}  \rho_{A|(BC)} V_{AB}^\dag  \\  &  + V_{AC}  \rho_{A|(BC)} V_{AC}^\dag \big),
\end{align}
where $V_{AB},V_{AC}$ refer to appropriate permutations. Here $\rho_{A|(BC)}$ stands for a PPT state with respect to partition $A|BC$, which additionally remains invariant under exchange of system $B$ and $C$. The structure of permutationally invariant states implies that the two qubits $BC$ couple to a spin-$1$ system, given by the symmetric subspace $Sym(BC)$ spanned by $\ket{11},\ket{\psi^+}=(\ket{01}+{10}))/\sqrt{2},\ket{00}$, and the spin-$0$ antisymmetric part $\ket{\psi^-}=(\ket{01}-\ket{10})/\sqrt{2}$. Thus this state (as well as its partial transposition with respect to $A$) can be decomposed into two parts as 
\begin{equation}
\label{eq:help1}
\rho_{A|(BC)} = q  \sigma_{A|Sym(BC)} + (1-q) \omega_{A}\otimes\ket{\psi^-}\bra{\psi^-}.
\end{equation}
with $\sigma_{A|Sym(BC)}$ being PPT. Since the symmetric subspace is three dimensional, this state is effectively a qubit-qutrit system for which PPT is equivalent to separability~\cite{horodecki96b}. Hence the state of Eq.~(\ref{eq:help1}) is separable and consequently the PI PPT mixture biseparable.
\end{proof}

This Observation complements the results of Ref.~\cite{guehne11a}, where an analogue result was shown for states with a different symmetry, namely graph-diagonal states of three and four qubits.

\section{Details of the SDP}\label{sec_sdpdetails}
Via the simplified form of a PPT mixture of a generic PI state the corresponding SDP can now be formulated as
\begin{align}
\label{SDPPI}
{\rm min}&~~-s \\ 
\nonumber
{\rm s.t.}&~~\rho_{(1\dots N)}=\sum\limits_{\substack{k=1}}^{N'/2}  {\left[Q_{(1\dots k)|(k+1\dots N)}\right]}_{PI},\\ 
\nonumber
&\:\begin{array}{c}
Q_{(1...k)|(k+1...N)}\geq s\mathbb{1},\\
Q_{(1...k)|(k+1...N)}^{T_{1\dots k}}\geq s\mathbb{1}, \end{array} ~ \mbox{ for all } k.
\end{align}
Apart from the polynomial scaling of the number of parameters involved in the SDP (c.f. Observation 2), it is also crucial to guarantee that the effort needed to verify the constraints scales also polynomially. For this, it is very important that the parameters are organized in blocks for if they were spread throughout the matrix in an unstructured way, checking its positivity would still be exponentially hard. Thus, in order not to deal with matrices of size $4^N$ we want to write the SDP constraints as constraints on smaller blocks that constitute the matrices involved in the program. It can be seen from Eq.~\eqref{decomposition} that the matrices $Q_{(1...k)|(k+1...N)}$, although not fully permutationally invariant, must also have a block structure, in a suitably chosen basis. We know from Sec.~\ref{sec_PIstates} that the operator basis elements $\sigma^{\alpha}_{(1...k)}$ and $\sigma^{\beta}_{(k+1...N)}$ have the block structure,
\begin{align}
\sigma^{\alpha}_{(1...k)}&=\overset{k/2}{\underset{j={j_k}_{min}}{\bigoplus}} C^{\alpha}_{j_{k}}\otimes\mathbb{1}_{\mathcal{K}_{j_{k}}},\\
\sigma^{\beta}_{(k+1...N)}&=\overset{\bar{k}/2}{\underset{j={j_{\bar{k}}}_{min}}{\bigoplus}} D^{\beta}_{j_{\bar{k}}}\otimes\mathbb{1}_{\mathcal{K}_{j_{\bar{k}}}}. 
\end{align}
From Eq.~\eqref{decomposition} it then follows
\begin{align}
& \!\!\!\!\!\!\!\!\!\!\!\!\!\!\!
Q_{(1...k)|(k+1...N)}
\nonumber \\
&=\sum_{\alpha,\beta}c_{\alpha\beta}^k\underset{j_k,j_{\bar{k}}}{\bigoplus}\left(C^{\alpha}_{j_{k}}\otimes D^{\beta}_{j_{\bar{k}}}\right)\otimes\left(\mathbb{1}_{\mathcal{K}_{j_{k}}}\otimes\mathbb{1}_{\mathcal{K}_{j_{\bar{k}}}}\right)
\nonumber\\
&\cong\underset{j_k,j_{\bar{k}}}{\bigoplus}B^k_{j_k,j_{\bar{k}}}\otimes\left(\mathbb{1}_{\mathcal{K}_{j_{k}}\otimes\mathcal{K}_{j_{\bar{k}}}}\right),
\label{tensor_blocks}
\end{align}
where we defined $B^k_{j_k,j_{\bar{k}}}=\sum_{\alpha,\beta}c_{\alpha\beta}^k C^{\alpha}_{j_{k}}\otimes D^{\beta}_{j_{\bar{k}}}$. The structure of Eq.~\eqref{tensor_blocks} is similar to Eq.~\eqref{eq:block_structure} which shows that, in a suitably ordered basis, the operator $Q_{(1\dots k)|(k+1\dots N)}$ is also block diagonal. Each block $B^k_{j_k,j_{\bar{k}}}$ has dimension $(2j_k+1)(2j_{\bar{k}}+1)$ and appears exactly $\dim(\mathcal{K}_{j_k})\dim(\mathcal{K}_{j_{\bar{k}}})$ times in the main diagonal. In this way, the SDP needs to store only the different blocks and keep track of how many times each block appears.
 
The task now is to define the SDP in terms of the these blocks $B^k_{j_k,j_{\bar{k}}}$ which constitute $Q_{(1...k)|(k+1...N)}$ and the corresponding blocks $B_j$ of the given PI state $\rho_{(1\dots N)}$. This task is direct for the matrix inequality constraints which translate to corresponding matrix inequalities
\begin{equation}
B^k_{j_k,j_{\bar{k}}}\geq s\mathbb{1},\hspace{6pt}({B^m_{j_k,j_{\bar{k}}}})^{T_{1\dots k}}\geq s\mathbb{1}
\end{equation}
for all blocks. Note that the operation of partial transposition is easy to implement, again due to the tensor product structure of Eq.~\eqref{decomposition}, because
\begin{equation}
Q_{(1\dots k)|(k+1\dots N)}^{T_{1\dots k}} =\sum_{\alpha\beta}c_{\alpha \beta} [\sigma^{\alpha}_{(1...k)}]^T\otimes\sigma^{\beta}_{(k+1...N)}
\end{equation} 
where $T$ denotes the usual transposition.
 
The implementation of the equality constraint, though, requires more care. We know that ${\left[Q_{(1\dots k)|(k+1\dots N)}\right]}_{PI}$ is block diagonal in the coupled spin basis of all $N$ particles $\ket{j,m,\alpha_j}$, as the original PI state, but the basis in which each $Q_{(1\dots k)|(k+1\dots N)}$ is block diagonal is a different one. In fact, it is diagonal in the basis $\ket{j_k,m_k,\alpha_{j_k};j_{\bar{k}},m_{\bar{k}},\alpha_{j_{\bar k}}} = \ket{j_k,m_k,\alpha_{j_k}} \otimes \ket{j'_{\bar{k}},m'_{\bar{k}},\alpha_{j_{\bar k}}}$. In general, the basis transformation between two spins $j_k,j_{\bar k}$ to a combined total spin $j$ is given by the Clebsch-Gordan coefficients
\begin{equation}\label{eq:clebschgordan}
\ket{j,m}\!=\!\!\!\!\!\sum\limits_{\substack{-j_k\leq m_k\leq j_k\\-j_{\bar{k}}\leq m_{\bar{k}}\leq j_{\bar{k}}}}\!\!\!\!\!\!\!\!\! \braket{j_k,m_k;j_{\bar{k}},m_{\bar{k}}|j,m}
\ket{j_k,m_k;j_{\bar{k}},m_{\bar{k}}}.
\end{equation}
This transformation holds for each spin, irrespective of the multiplicative spaces. However, via the multiplicative spaces one keeps track of how many spins $j_k$ of system $1\dots k$, and how many spins $j_{\bar k}$ of system $k+1 \ldots N$ couple with each other. To compute the resulting blocks $B_j$ of $\sum_k{\left[Q_{(1\dots k)|(k+1\dots N)}\right]}_{PI}$ the procedure is hence as follows:
\begin{itemize}
\item From the blocks $B_{j_k j_{\bar k}}^k$ of $Q_{(1\dots k)|(k+1\dots N)}$, typically given in the basis $\ket{j_k,m_k;j_{\bar k},m_{\bar k}}$, one first computes their contribution to each total spin $j$, using $\ket{j,m}$, via the transformation of Eq.~\eqref{eq:clebschgordan}. This result is denoted as $B^{k,j}_{j_k j_{\bar k}}$. Note that one only gets a non-trivial matrix if the two individual spins can form at all a total spin $j$, \ie, $|j_k - j_{\bar k}| \leq j \leq j_k + j_{\bar k}.$  
\item Afterwards one performs the average over all possibilities, more precisely, 
\begin{equation}
B^j_k = \sum_{j_k j_{\bar k}} \frac{\dim(\mathcal{K}_{j_k})\dim(\mathcal{K}_{j_{\bar k}})}{\dim(\mathcal{K}_{j})} B^{k,j}_{j_k j_{\bar k}}.
\end{equation}
Here $\dim(\mathcal{K}_{j_k})\dim(\mathcal{K}_{j_{\bar k}})$ is the number of spins $j_{k},j_{\bar k}$ (which couple to a spin $j$) in the original operator.
\end{itemize}
This way, after summing over $k$, the RHS of the equality constraint of Eq.~\eqref{SDPPI} is computed. Due to its symmetry, $Q_{(1\dots k)|(k+1\dots N)}$ has a polynomial number of parameters, so the basis transformation requires the computation of only a polynomial number of Clebsch-Gordan coefficients. This discussion leads to the final observation of this work: 

\begin{observation}
The SDP to detect genuine multipartite entanglement of PI states of $N$ qubits via the concept of PPT mixtures can be formulated in terms of $\mathcal{O}(N^3)$ matrices whose size is at most $\mathcal{O}(N^2)$.
\end{observation}
Since the number of different blocks of a general PI operator on $k$ parties is $\mathcal{O}(k)$, each operator $Q_{(1\dots k)|(k+1\dots N)}$ has about $\mathcal{O}(N^2)$ different blocks. Since we need roughly $N/2$ of these operators for a PI PPT mixture we have of $\mathcal{O}(N^3)$ operators $B^k_{j_k,j_{\bar{k}}}$ in total. For each of these operators we need to check two matrix inequalities. Furthermore, the biggest of the blocks is the one with $k=N^\prime/2$ and $j_k=N^\prime/4, j_{\bar k} = (2N-N^\prime)/4$, hence the maximal dimension is $\mathcal{O}(N^2)$.

\section{Examples}\label{sec:V}
In this section, we present some examples for the application of the SDP, illustrating the strength of the PPT mixtures. Our first example is the calculation of the white noise tolerance for Dicke states which we compare to the method of Ref.~\cite{huber_dicke}. In our second example, we consider a mixtures of a GHZ state, a W state and white noise and compare the detection range with the ones of Refs.~\cite{guehneseevinck, huber_ghzwid}, which were the best known results so far for this class of states. In both cases, we achieve significantly improved results. However, since our method is based on a numerical approach, it is limited by the memory of the computer and we could only run it for states of at most $10$ qubits with the first prototype of the program. Note that we did not further optimize the algorithm for these special kind of states. In contrast the criteria of Refs.~\cite{guehneseevinck, huber_ghzwid} are analytic and can therefore be applied to arbitrary qubit numbers. 
Further criteria from the approach of PPT mixtures, which similarly rely on analytic estimates will be discussed elsewhere \cite{marcel2013}. 

\emph{Example 1.---}Dicke states have first been studied in the context of light emission from a cloud of atoms \cite{dicke_states} and have been prepared in many experiments \cite{haeffner05, weinfurter06}. The symmetric $N$-qubit Dicke state with $k$ excitations is defined as the superposition of all basis states with $k$ excitations, 
\begin{equation}
\ket{D_{N,k}}=\frac{1}{\sqrt{{N\choose k}}}
\sum_{\pi\in S_N}V(\pi)\ket{1}^{\otimes k}\otimes \ket{0}^{\otimes N-k}
\end{equation}
and is therefore a permutationally invariant state. For example, the four-qubit Dicke state with two excitations is given by $\ket{D_{4,2}}=(\ket{0011}+\ket{0101}+\ket{0110}+\ket{1001}+\ket{1010}+\ket{1100})/\sqrt{6}.$

\begin{figure}[t!]\label{fig:noisetol_dicke}
\centering
\includegraphics[scale=0.35,angle=-90,trim=0cm 3.cm 3cm 0cm]{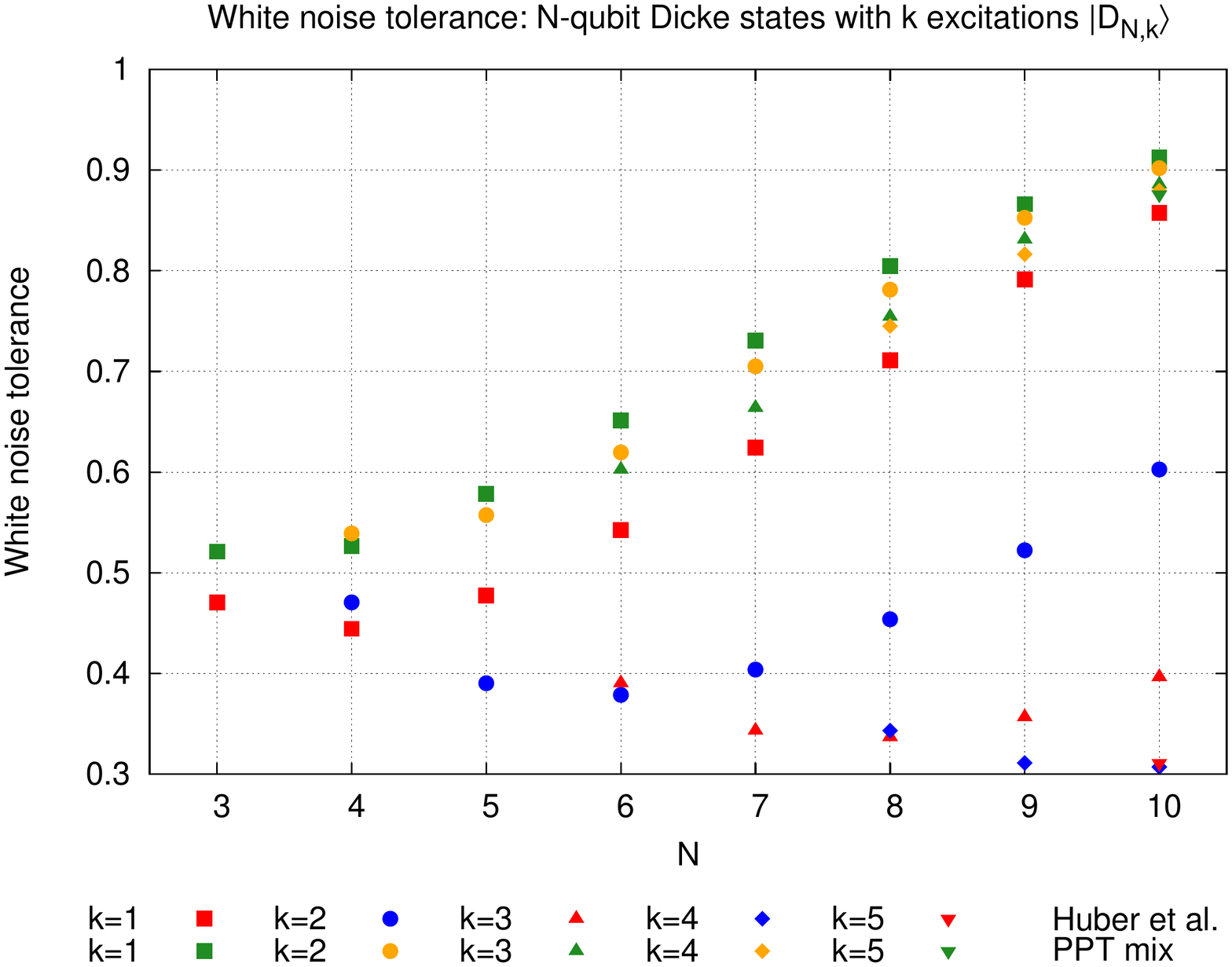}
\vspace{0.9cm}
\caption{Comparison between the PPT mixture criterion (filled symbols) and the criterion of Huber et al. \cite{huber_dicke} (empty symbols) of the white noise tolerance for Dicke states $\ket{D_{N,k}}$. We show the white noise tolerances for $N$ up to $10$ qubits and $k$ up to $N/2$. The PPT mixture criterion is more robust to white noise and in some cases, the difference of the white noise tolerance between both criteria is very significant reaching values larger than 40\%.}
\end{figure}

We computed the white noise tolerance for Dicke states of up to $10$ qubits and up to $N/2$ excitations comparing it to the criterion of Ref.~\cite{huber_dicke}. From Fig.~3 we see that the PPT mixture criterion is always more robust to noise and the difference is more significant for larger number of qubits and excitations. The improvement reaches values larger than 40\%, which should prove itself useful for current Dicke state experiments.

\emph{Example 2.---}Other well known states, which are also invariant under permutations are the GHZ and W state. The GHZ state is defined as
\begin{equation}
\ket{GHZ_N}=\dfrac{1}{\sqrt{2}}\left(\ket{0}^{\otimes N}+\ket{1}^{\otimes N}\right),
\end{equation}
while the W state is the Dicke state with one excitation,
\begin{equation}
\ket{W_N}=\dfrac{1}{\sqrt{N}}\left(\ket{10...0}+\ket{01...0}+...+\ket{0...01}\right).
\end{equation}
In this example, we consider the following $N$-qubit states
\begin{equation}\label{eq:ghzwid}
\rho(p_1,p_2)=p_1\rho_{GHZ_N}+p_2 \rho_{W_N}+(1-p_1-p_2)\dfrac{\mathbb{1}}{2^N},
\end{equation}
with $\rho_{GHZ_N}=\ket{GHZ_N}\bra{GHZ_N}$, $\rho_{W_N}=\ket{W_N}\bra{W_N}$. These states are a convex combination of a GHZ state, a W state and white noise. 

Such states can be represented by a point in a two-dimensional plane whose coordinates are given by $p_1$ and $p_2$, as shown in Fig.~\ref{fig:ghzwid3} and Fig.~\ref{fig:ghzwid8}. Of course, only certain pairs of $p_1,p_2$ correspond to valid quantum states, hence, only the lower triangle shown in these figures describe actual quantum states of this class. In these figures, which corresponds to three and eight qubits respectively, we furthermore show the set of states which is detected by the PPT mixture criterion in comparison to the method developed in Refs.~\cite{guehneseevinck, huber_ghzwid}. In both cases the set of states detected by the PPT mixtures is much larger than the one verified by the aforementioned criteria and this difference grows with the number of qubits. In Fig.~\ref{fig:ghzwid3}, in fact, the PPT mixtures criterion is optimal, cf. Sec.~\ref{necsuf3qubits}, so the states that have a PPT mixture are biseparable.

\begin{figure}[t]
\centering
\includegraphics[scale=0.4, trim=4cm 2cm 0cm 0cm]{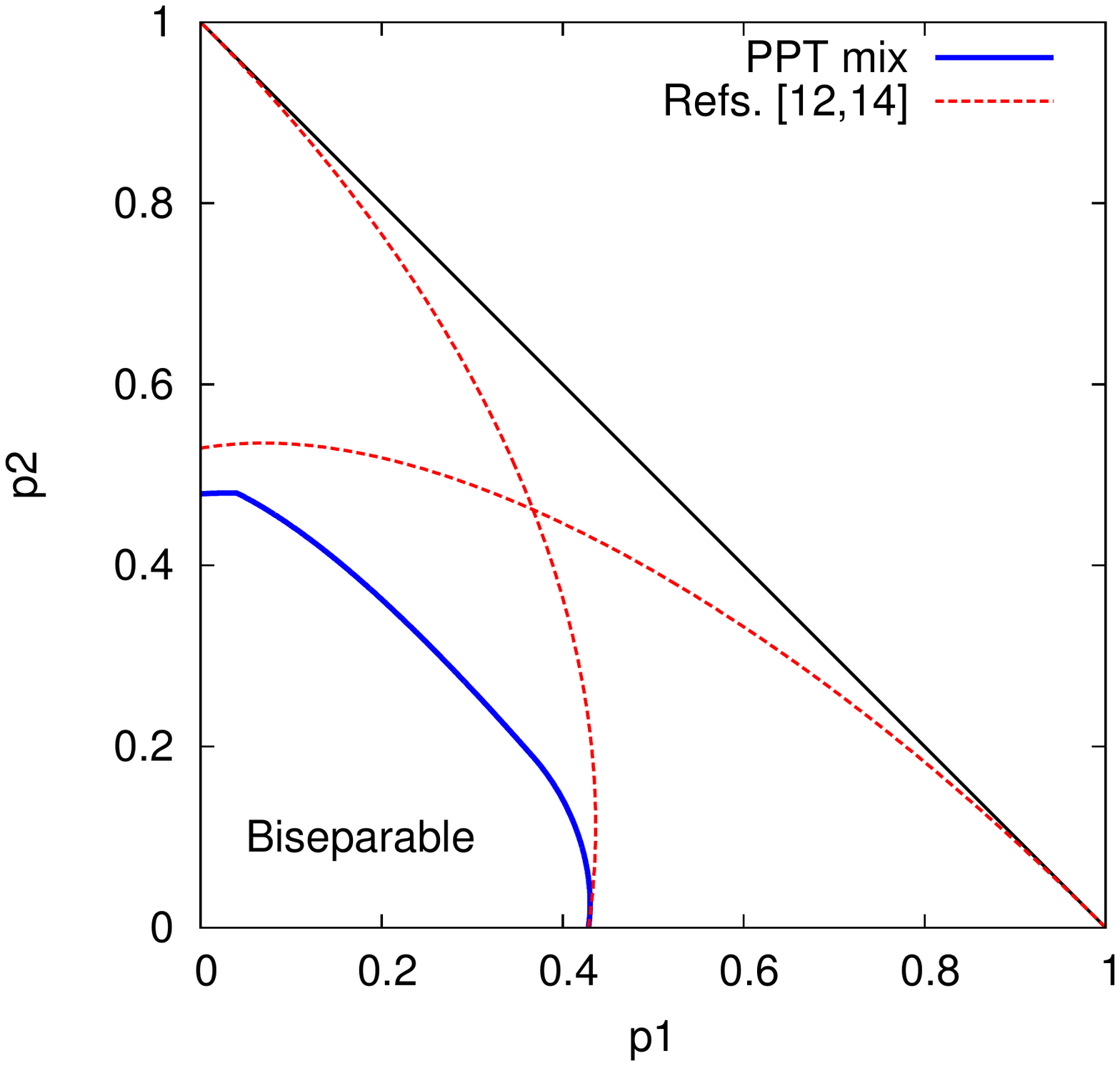}
\caption{Here, we show which three-qubit states (lower triangular) as defined by Eq.~(\ref{eq:ghzwid}) are detected by the PPT mixtures criterion as genuinely multipartite entangled (area outside the solid line). We compare it to the criteria presented in Refs.~\cite{guehneseevinck, huber_ghzwid}, consisting of two inequalities, one optimized for the GHZ state and the other for W state (dashed lines). The PPT mixtures criterion represents a significant improvement and is optimal for three-qubit PI states, so the states inside the solid line are biseparable.}
\label{fig:ghzwid3}
\end{figure}
\begin{figure}[h!]
\includegraphics[scale=0.4, angle=-90, trim=0cm 4cm 2cm 0cm]{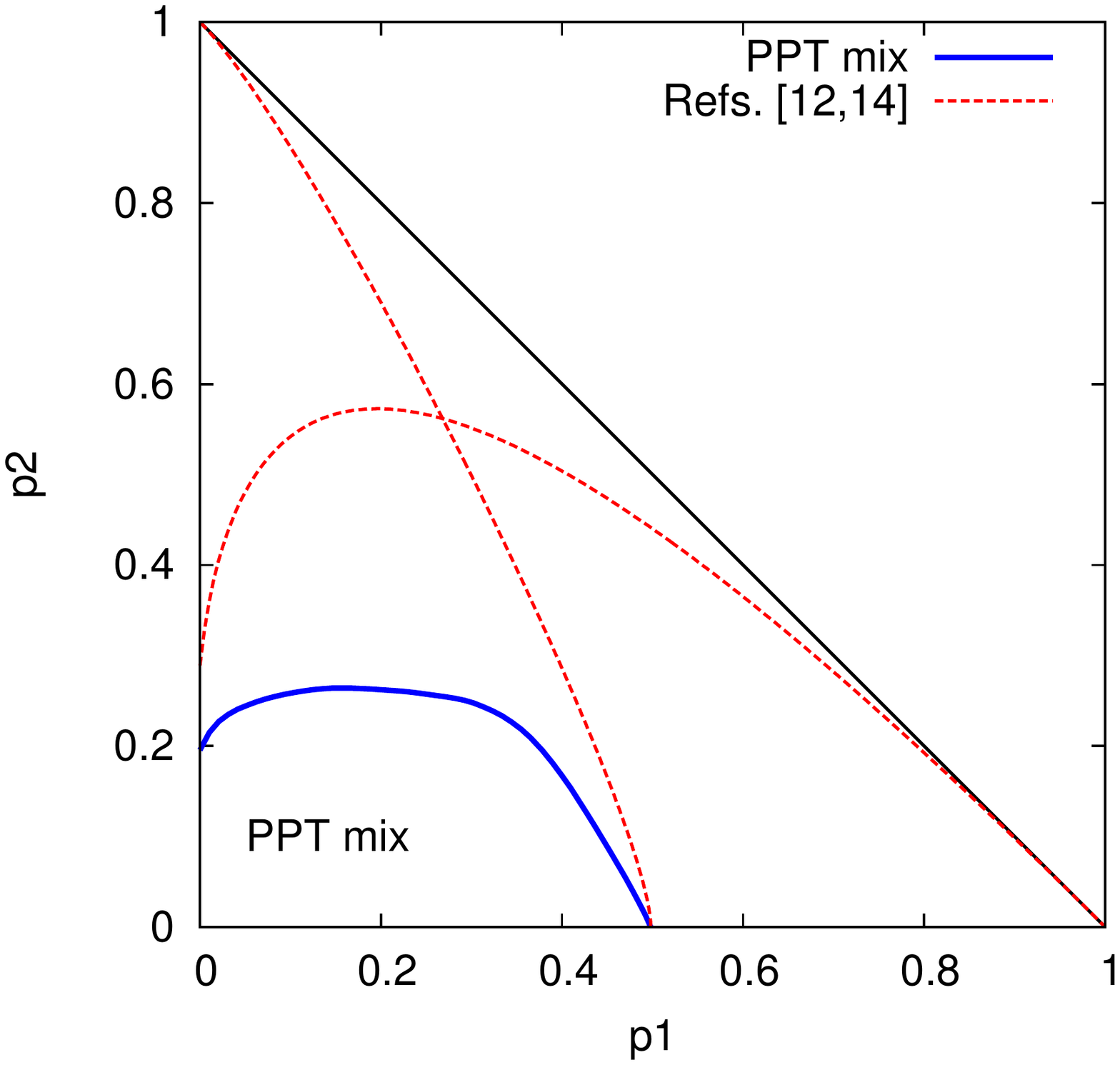}
\caption{The same as Fig.~\ref{fig:ghzwid3}, but for eight-qubit states. The PPT mixture criterion represents a very significant improvement, but note that here we cannot prove that the PPT mixture criterion is necessary and sufficient for biseparability.}
\label{fig:ghzwid8}
\end{figure}
\section{Conclusions} \label{sec:VI}
In this paper we tailored the detection of genuine multipartite entanglement via PPT mixtures for permutationally invariant states. In contrast to a generic $N$-qubit state, where the question of characterizing  PPT mixtures scales exponentially with the number of particles, our optimization for this special class of states only requires a polynomial scaling of the order $\mathcal{O}(N^7)$. This was possible by deriving a more restricted but still general form of a PPT mixture using the additional symmetry of the state. Via this method, we were able to analyze more rigorously the entanglement for system sizes where the original numerical PPT mixture method would fail.  

In addition, we have  shown that the criterion of PPT mixtures completely solves the question of genuine multipartite entanglement for permutationally invariant three-qubit states. This furthermore supports the conjecture, motivated in Ref.~\cite{guehne11a}, that PPT mixtures are necessary and sufficient for biseparability of three qubits. We leave this open for further discussion. 

On more general ground, we believe that the development or optimization of existing analysis tools to larger-scale systems is a mandatory step for a well-grounded investigation of the properties of systems with many particles and its experimental implementation. This should help to close the gap between methods for small system playgrounds and the really interesting system sizes that could deserve the term quantum computer at some time.

\section{Acknowledgements}
We thank Marcel Bergmann for discussions and Marcus Huber for providing the curves for the comparison with the PPT mixture criterion. This work has been supported by the EU (Marie Curie CIG 293993/ENFOQI) and the BMBF (Chist-Era Project QUASAR). LN acknowledges also the support from project IT-PQuantum, as well as from Funda\c{c}\~{a}o para a Ci\^{e}ncia e a Tecnologia (Portugal), namely through programme POC\-TI/PO\-CI/PT\-DC and project PTDC/EEA-TEL/103402/2008 QuantPrivTel, partially funded by FEDER (EU), and from the European Union's Seventh Framework Programme (FP7/2007-2013) under grant agreement nb. 318287 project LANDAUER.

\bibliographystyle{apsrev}
\bibliography{biblio}

\begin{thebibliography}{37}
\bibitem[{ppt()}]{pptmixer}
\bibinfo{note}{The program {\tt PPTmixer} is available at {\tt
  mathworks.com/matlabcentral/fileexchange/30968}}.
\end{thebibliography}

\end{document}